\documentclass[letterpaper,10pt]{article}

\usepackage{etoolbox}
\newtoggle{article}
\toggletrue{article}

\usepackage{amsthm,amsmath,amssymb,amsfonts}
\usepackage{mathpazo}
\usepackage{adjustbox}
\usepackage{placeins}
\usepackage{natbib}
\usepackage{authblk}
\def\cite{\citep}

\usepackage[svgnames]{xcolor}
\usepackage{graphicx}

\usepackage[
  pdfauthor={Mohamed Khochtali, Daniel S. Roche, Xisen Tian},
  pdftitle={Parallel sparse interpolation using small primes},
  colorlinks,
  linkcolor=DarkBlue,
  citecolor=DarkGreen,
  urlcolor=DarkBlue,
  bookmarksnumbered,
  pdfpagelabels=true,
]{hyperref}
\usepackage{doi}

\usepackage[vlined,ruled,boxed,linesnumbered]{algorithm2e}
\DontPrintSemicolon

\usepackage[subject={TODO}]{pdfcomment}

\newtheorem{theorem}{Theorem}
\newtheorem{claim}[theorem]{Conjecture}
\newcommand{\ZZ}{\ensuremath{\mathbb{Z}}}
\newcommand{\softoh}{\ensuremath{\widetilde{O}}}

\title{Parallel sparse interpolation\\ using small primes}

\author{Mohamed Khochtali}
\author{Daniel S.\ Roche}
\author{Xisen Tian}
\affil{\small Computer Science Department\\United States Naval Academy\\
Annapolis, Maryland, USA}

\begin{document}

\maketitle

\begin{abstract}
  To interpolate a supersparse polynomial with integer coefficients, two
  alternative approaches are the Prony-based ``big prime'' technique,
  which acts over a single large finite field, or the more
  recently-proposed ``small primes'' technique, which reduces the unknown
  sparse polynomial to many low-degree dense polynomials.
  While the latter technique
  has not yet reached the same theoretical efficiency as Prony-based
  methods, it has an obvious potential for parallelization. We present a
  heuristic ``small primes'' interpolation algorithm and report on a
  low-level C implementation using FLINT and MPI.
\end{abstract}

\section{Introduction}

Given a way to evaluate or \emph{sample} an unknown function or
procedure, interpolation is the fundamental and important problem of
recovering a formula which accurately and completely describes that
unknown function. As discovering an arbitrary unknown function from a
finite set of evaluations with any reliability would be impossible, some
constraints on the size and form of the output are inevitably required.

Here we consider the problem of \emph{sparse polynomial interpolation},
in which we are guaranteed that the unknown function is a multivariate
polynomial with bounded degree. Sparse interpolation algorithms date to
the 18th century, but have been the focus of considerable recent work in
numeric and symbolic computation,
with applications ranging from power consumption in medical devices, 
to reducing intermediate expression swell in mathematical computations
\cite{CL08,Kal10a,KLY11,BCK12}.

Specifically, this work focuses on algorithms to interpolate an unknown
\emph{supersparse} polynomial with integer coefficients, which make
efficient use of modern \emph{parallel} computing hardware. Focusing on
the ``supersparse'' (a.k.a. ``lacunary'') case means that our running
time will be in terms of the number of variables, number of nonzero
terms, and the logarithm of the output degree. 

The first algorithms to solve this problem in polynomial-time were based
on the exponential sums technique of Prony, and can efficiently solve
the integer problem by working in a large finite field modulo a single
``big prime.'' A number of theoretical improvements and practical
implementations have been made in this vein, including work on fast
parallel implementations \cite{JM10,HL15}.

We consider another type of approach, whereby the unknown sparse
integer polynomial is reduced in degree modulo many small primes. This
technique, first used by Garg and Schost to avoid the need for discrete
logarithm computations in arbitrary finite fields, can also be applied
to the integer polynomial case. While it cannot yet match the
theoretical efficiency of the big primes algorithms, we will show that the
``small primes'' method is very effectively parallelized. Furthermore, we
develop a practical heuristic version of this method which reduces
further the
size and number of primes required based on experimental results rather
than on the theoretical worst-case bounds.

\subsection{Related work}

While it would be impossible to list all of the related work on sparse
interpolation, we will mention some of the most recent results which are
most closely connected to the current study, and which may provide the
reader with useful background.

The now-classical approach to sparse interpolation is variously credited
to Prony, Blahut \cite{Bla79}, or Ben-Or and Tiwari \cite{BT88}. Only
the last of these considered explicitly the case of integer
coefficients, but all share the key property of requiring the minimal
number $O(T)$ of evaluations in order to recover a polynomial with 
$T$ nonzero terms.

We refer to these approaches as ``big prime'' techniques, as the more
modern variants \cite{KY89,Kal10a} adapt to the case of integer
coefficients by choosing carefully a single large modulus, then perform
the interpolation over a finite field in order to avoid exponential
growth in the bit-length of evaluations.

The approach which we take in this work is based on that of Garg and
Schost \cite{GS09}, who developed the first polynomial-time supersparse
interpolation algorithm over an arbitrary finite field. By reducing the
unknown polynomial modulo $(z^p-1)$, the full coefficients
are discovered immediately, but the exponents are only discovered after
repeating for multiple values of $p$. This ``small primes'' approach,
described in more details below, has the considerable advantage of
relying only on low-degree, dense polynomial arithmetic.

There are, of course, other sparse interpolation methods which do not
fit nicely into this big/small prime characterization. Most notable are
Zippel's algorithm and hybrid variants of it \cite{Zip90,KL03}, the
symbolic-numeric method of \cite{Man95}, and the Newton-Hensel
lifting approach of \cite{AKP06}.

We point out some recent work on efficient implementations which are of
particular interest to the current study. In \cite{JM10}, a new variant
of the big prime approach is developed which can be performed variable
by variable, in parallel. More recently, \cite{HL15} investigated a
number of tricks and techniques towards practical, efficient sparse
interpolation, and posed some new benchmark problems. Their methods are
also based on the big prime approach; to our knowledge there has been no
reported implementation work on the small primes technique other than
the numerical interpolation code reported in \cite{GR11a}.

\subsection{Our contributions}

Suppose $f\in\ZZ[x_1, x_2,\ldots x_n]$ is an unknown polynomial in $n$
variables, with partial degrees less than $D$ in each variable,
at most $T$ nonzero terms, and coefficients less than $H$ in absolute
value. Then $f$ can be written as
\[f = \sum_{i=1}^T c_i x_1^{e_{i1}}x_2^{e_{i2}}\cdots{}x_n^{e_{in}},\]
where each exponent $e_{ij}< D$ and each $|c_i| < H$.

Given a way to evaluate $f(\theta_1,\ldots,\theta_n)\bmod q$, for any
modulus $q\in\ZZ$ and $n$-tuple of evaluation points
$(\theta_1,\ldots,\theta_n) \in (\ZZ/q\ZZ)^n$,
the \emph{sparse integer polynomial interpolation problem} is to
determine the coefficients $c_i$ and exponent tuples 
$(e_{i1},\ldots,e_{in})$, for each $1\le i\le T$.

We will actually consider a slight relaxation of this problem, wherein
evaluations are of the form
$$(q,p,d_1,\ldots,d_n)
  \mapsto f(z^{d_1},\ldots,z^{d_n}) \bmod (z^p-1) \in (\ZZ/q\ZZ)[z].$$
That is, the $n$ variables are replaced by a single indeterminate $z$,
all coefficients are reduced modulo $q$, and the exponents are
reduced modulo $p$. The reduction modulo $(z^p-1)$ 
is possible without affecting the overall
complexity whenever the unknown polynomial $f$ is given as a
straight-line program or algebraic circuit, or if the prime $q$ is
chosen so that $\ZZ/q\ZZ$ has a $p$th root of unity.

An algorithm for this problem is said to handle the \emph{supersparse} case
if it requires a number of evaluations and running time which are
polynomial in $n$, $T$, $\log D$, and $\log H$. This corresponds to the size of
the sparse representation of $f$ as a list of coefficient-exponent
tuples, which requires
$O(nT\log D + T\log H)$
bits in memory.

We present a randomized algorithm for the sparse integer polynomial
interpolation problem, derived from the existing literature, whose
running time%
\footnote{
  Here and throughout we use the \emph{soft-oh} notation in order to
  simplify the stated running times: a running time is said to be 
  $\softoh(\phi)$ for some function $\phi$ if and only if it is
  $O(\phi (\log \phi)^{O(1)})$.
} for provable correctness is
\begin{equation}\label{eqn:rtg}
\softoh\left(
\left(1 + \tfrac{1}{m}n\log D\right)
n T \log D (\log D +\log H)\right),
\end{equation}
where $m \ge 1$ is the number of parallel processors available for the
task. Furthermore, we demonstrate a heuristic variant on this algorithm,
which works well in our experimental testing, and reduces the running
time further to
\begin{equation}\label{eqn:rth}
\softoh\left(
  n\log D\log H
  + \left(1 + \tfrac{1}{m}n\log D\right) T \log H
  \right)
\end{equation}
in the typical case that $n$ and $\log D$ are both $O(T\log H)$.

We have implemented this heuristic approach using the C library FLINT
for dense polynomial arithmetic and MPI for parallelization. Our
experiments demonstrate the smallest effective settings for the
parameters in our heuristic approach. With those parameters, we show
that the heuristic method is competitive with the state of the art in
the single-processor setting, and that its running
time scales well with increasing numbers of parallel processors.

Specifically, our contributions are:
\begin{enumerate}
  \item A sparse interpolation algorithm whose potential parallel
    speedup is $O(n\log D)$, compared with the $O(n)$ parallel speedup
    that has been shown in previous work for other sparse interpolation algorithms.
  \item A heuristic variant of our algorithm, which is demonstrated to
    be effective
    on our (limited) random experiments, that brings the
    running time complexity of the small primes approach to be
    competitive with that of the big prime approach.
  \item An efficient C implementation of our interpolation algorithm
    which demonstrates its competitiveness on a standard benchmark
    problem.
\end{enumerate}

\subsection{Organization of the paper}

We first outline the two main classes of existing approaches to
supersparse integer polynomial interpolation, which we call the ``big
prime'' and the ``small primes'' methods, in Section~\ref{sec:exist}. We
then present our own heuristic method in Section~\ref{sec:ouralg}, 
which is based
on previously-known ``small primes'' algorithms, but goes beyond
theoretical worst-case bounds on the sizes needed in order to further
improve the efficiency.

Section~\ref{sec:impl} reports the details of our parallel
implementation of this heuristic method, and Section~\ref{sec:exper}
presents the preliminary experimental results which demonstrate the
efficacy of this approach. Finally, we state some conclusions and directions for
further investigation in Section~\ref{sec:conc}.

\section{Existing algorithms for supersparse interpolation}
\label{sec:exist}

\subsection{``Big prime'' methods}

The original sparse interpolation 
algorithm of \cite{BT88} used evaluations at powers of the first $n$
prime numbers along with Prony's method to deterministically recover the
nonzero integer coefficients and exponents of an unknown polynomial.
This cannot be considered a ``supersparse'' algorithm, as it
performs the evaluations over the integers directly and
requires working with integers with more than $D$ bits.

However, it was soon recognized that, by choosing a single large prime
$p \ge D^n$, with $p-1$ smooth%
\footnote{
  An integer is said to be \emph{smooth} if it has only small prime
  factors.
} so as to facilitate discrete logarithms,
a supersparse integer polynomial could be interpolated efficiently
modulo $p$ \cite{KLW90,Kal10a}. The basic steps of this approach are as
follows:

\begin{enumerate}
  \item Choose prime $p$ so that $(p-1)$ has a large ``smooth''
    factor greater than $D^n$ and let $\theta$ be a primitive element
    modulo $p$.
  \item For $i=0,1,2,\ldots,2T-1$, evaluate
    \[v_i = f(\theta^i, \theta^{Di}, \ldots, \theta^{D^{n-1}i}) \bmod p.\]
  \item Compute the minimum polynomial $\Gamma\in\ZZ_p[z]$ of the 
    sequence $(v_i)_{i\ge 0}$ with the Berlekamp-Massey algorithm.
  \item Factor $\Gamma$ over $\ZZ_p[z]$; each root of $\Gamma$ can be
    written as
    $\theta^{e_1 + e_2 D + \cdots + e_n D^{n-1}}$, corresponding to
    a single term $c x_1^{e_{i1}}\cdots x_n^{e_{in}}$ of $f$.
  \item Compute $T$ discrete logarithms of the roots, and then the
    $D$-adic expansion of each one, to discover the actual exponents
    $(e_{i1},\ldots,e_{in})$, for $1\le i\le T$.
  \item Once the exponents are known, the coefficients can be computed
    from the evaluations $v_i$ by solving a transposed
    Vandermonde system of dimension $T$.
\end{enumerate}

The two steps which can be trivially parallelized are the evaluations
and discrete log computations on steps 2 and 5. However, the dominating
cost in the complexity is factoring a polynomial over $\ZZ_p[x]$ on step
4. This is also confirmed to be the dominating cost in practice by
\cite{HL15}, and it is not clear how to efficiently parallelize the
factorization. Using the fastest known algorithms, the running time of
this step is $\softoh(T\log^2 p)$ \cite{GHL15}, which is at least 
$\softoh(n^2 T \log^2 D)$ bit operations from the condition
$p > D^n$.

In the description above, the evaluations at powers of $\theta$ on step 2 amount to a
Kronecker substitution from multivariate to univariate. Of note is
the algorithm of \cite{JM10}, which uses a different approach than
Kronecker substitution in order to work one variable at a time, gaining
a potential $n$-fold parallel speedup.

\subsection{``Small primes'' methods}

The algorithm described above works the same over an arbitrary finite
field, except that the discrete logarithms required in step 5 cannot be
performed in polynomial-time in general. This difficulty was first
overcome in \cite{GS09}, where the idea is as follows:

\begin{enumerate}
  \item Choose a series of small primes $p_1, p_2, \ldots$
  \item Apply the Kronecker substitution and for each $i=1,2,\ldots$
    compute $f_i = f(z,z^D,\ldots,z^{D^{n-1}}) \bmod (x^{p_i}-1)$.
  \item Each $f_i$ of maximal sparsity contains all the coefficients of
    $f$, and all the exponents modulo $p_i$. Collect sufficiently many
    modular images of the exponents in order to recover the full
    exponents over $\ZZ$.
  \item Recover the multivariate exponents by $D$-adic expansion of each
    univariate exponent, and use any $f_i$ of maximal sparsity to
    discover the coefficient of each term.
\end{enumerate}

There are two significant challenges of this approach. The first is
that, in reducing the polynomial modulo $x^{p_i}-1$, it is possible that
some exponents are equivalent modulo $p_i$, and then multiple terms in
the original polynomial will \emph{collide} to form a single term in
$f_i$. By choosing random primes whose values are roughly
$\softoh(T^2\log D)$, the probability of encountering any collisions can
be made arbitrarily small.

The second challenge is how to correlate the exponents from different
$f_i$'s in step 3, in order to recover the full exponents via Chinese
remaindering. The approach of \cite{GS09} was to compute an auxiliary
polynomial whose roots are the unknown exponents; however this
increases the overall running time to $\softoh(n^2 T^3 \log^2 D \log
H)$.

The technique of \emph{diversification} in \cite{GR11a} is another
randomization which chooses a random element $\alpha$ and interpolates
$f(\alpha z)$ instead of $f(z)$ itself. With high probability, the
``diversified'' polynomial $f(\alpha z)$ has distinct coefficients,
which can then be used to correlate the exponents in different $f_i$'s.
This avoids the factoring step and reduces the complexity by a factor of
$T$.

Subsequently, and separately, \cite{AGR13} showed how to allow the
magnitude of each $p_i$ to decrease by a factor of $T$, by allowing some
constant fraction of the terms in each $f_i$ to collide. These separate
approaches are combined and further improved in \cite{AGR15},
which in the typical case that
$n \ll \log D \ll \log H$, brings the theoretical complexity down to
$\softoh(nT \log^2 D \log H)$. Observe that this is competitive with the
``big primes'' algorithm, but could be slower by up to a factor of
$\log H$. 

The algorithm we describe next first shows how to effectively
parallelize this small primes approach, and then further reduces the
complexity through a heuristic argument. After the gains from
parallelization or from the heuristic improvement, the complexity of our
algorithm will be less than the ``big primes'' approach as well.

\section{Our parallel small primes algorithm}
\label{sec:ouralg}

\begin{procedure}[tb]
\caption{SparseInterp($f, n, T, D, H, k, \ell$)\label{proc:interp}}
\KwIn{Bounds $T, D, H$ for an $n$-variate sparse polynomial $f$ in $n$
variables, with $T$ nonzero terms, partial degrees less than $D$, and
integer coefficients less than $H$ in absolute value; and parameters
$k,\ell,Q\in\ZZ_{>0}$.}
\KwOut{A list of $t\le T$ coefficients $c_i$ and exponent tuples
$(e_{i1},\ldots,e_{in})$. If $k,\ell$ are sufficiently large, these
coefficients and exponents comprise the sparse representation of $f$
with high probability.}

$q \gets$ random prime in the range $[Q, 2Q]$ \;
$\alpha \gets$ random element of $(\ZZ/q\ZZ)^*$ \;
$(\alpha_0,\ldots,\alpha_{n-1}) \gets 
  (\alpha, \alpha^D\bmod q, \ldots, \alpha^{D^{n-1}}\bmod q)$ \;
$\lambda \gets kT$ \;
$\mu \gets \lceil (\ell n \lg D)/\lg \lambda \rceil$ \;
$L \gets$ thread-safe list of integer triples \;
\For{$i=1,2,\ldots,\mu$ in parallel}{
  $p_i \gets $ random prime in the range $[\lambda,2\lambda]$ \;
  $(D_0,\ldots,D_{n-1}) \gets (1, D\bmod p_i, \ldots, D^{n-1}\bmod p_i)$ 
    \label{step:Dpow} \;
  $f_i \gets f(\alpha_0 z^{D_0}, \ldots, \alpha_{n-1} z^{D_{n-1}})
    \bmod (z^{p_i}-1)
    \in (\ZZ/q\ZZ)[z]$ using dense polynomial arithmetic \label{step:eval} \;
  \For{$j=0,1,\ldots,p_i-1$}{
    $c_{ij} \gets $ coefficient of $z^j$ in $f_i$ \;
    \lIf{$c_{ij} \ne 0$}{Add $(c_{ij}, j, p_i)$ to $L$
    \label{step:growL}}
  }
}
Sort $L$ in parallel by the coefficients $c_{ij}$ in each triple \;
$F \gets$ thread-safe list of coefficient/exponent tuples \;
\ForEach{Unique coefficient $c$ in $L$ in parallel}{
  $(d_{c1},p_{c1}),\ldots,(d_{cu},p_{cu}) \gets$ 
    the $u\ge 1$ exponents and primes appearing with coefficient $c$
    in $L$ \;
  \If{$u \ge \mu/2$}{
    $u' \gets $ least integer s.t.\ $\prod_{1\le i \le u'} p_{ci} \ge D^n$ \;
    $E_c \gets$ least integer s.t.\ 
      $E_c \equiv d_{ci}\bmod p_{ci}$ for $0\le i\le u'$ via
      Chinese remaindering \;
    $c^* \gets c \alpha^{-E_c} \in \ZZ/q\ZZ$, stored as an integer in the range
      $[-q/2,q/2]$ \;
    $(e_1,\ldots,e_n) \gets $ $D$-adic expansion of $E_c$ \;
    Add $c^*$ and $(e_1,\ldots,e_n)$ to $F$\;
  }
}
Sort $F$ in parallel by the exponents and \textbf{return} the resulting sparse
polynomial
\end{procedure}

Our algorithm, which is detailed in procedure \ref{proc:interp},
is based on the reduction idea in \cite{GS09}, with the
diversification method introduced by \cite{GR11a} and the partial
collisions handling of \cite{AGR13}. It depends crucially on the
parameters $k$ and $\ell$, which in theory grow as $O(n\log D)$ and
$O(1)$, respectively, but according to our heuristic method can both be
treated as constants.

\subsection{Algorithm overview}

The main idea is first to choose a prime $q$ large enough to recover the
coefficients, and then to apply the Kronecker substitution so that we
are really interpolating a univariate polynomial $g(z)$
with coefficients in $\ZZ/q\ZZ$:
\[g = f(\alpha z,(\alpha z)^D,\ldots,(\alpha z)^{D^{n-1}})
= \sum_{i=1}^T c_i z^{E_i} \in (\ZZ/q\ZZ)[z].\]
Each term $c^*x_1^{e_1}\cdots x_n^{e_n}$ of the original polynomial $f$
maps uniquely to a term with exponent
$E_i = e_1 + e_2 D + \cdots + e_n D^{n-1}$ and coefficient
$c_i = c^*\alpha^{E_i} \bmod q$ in $g$.

The parameter $Q$ controls the size of the prime $q$, and so should
always be set larger than $2H$ in order to recover the full precision of
the coefficients. However, it may be necessary to set $Q$ even larger
than this when the height bound $H$ is very small. We comment that
choosing a random prime $q$ (rather than one with special properties, as
in the ``big primes'' algorithm) is important for the probabilistic
analysis below.

The evaluation phase of the algorithm computes the polynomial
$g$ modulo $(z^p-1)$ using
dense arithmetic, for many small primes $p$. The details of 
how this evaluation is performed
will depend on the particular application. In our
multivariate multiplication application below, the exponents of the
original multiplicands are reduced modulo $p$, followed by dense
polynomial arithmetic in the ring of polynomials modulo $z^p-1$. More
generally, if the unknown polynomial $f$ is given as a straight-line
program or arithmetic circuit, each operation in the circuit can be
computed over that ring $(\ZZ/q\ZZ)[z]/\langle z^p-1\rangle$, as in
\cite{GS09}. In our complexity analysis below, we assume a cost of
$\softoh(p\log q)$ for each evaluation on
this step, according to the cost of dense
degree-$p$ polynomial arithmetic over $\ZZ/q\ZZ$.

The next phase of the algorithm is to gather the images of each term,
discard those which appear infrequently (and thus were resulting from
collisions in the reduction modulo $z^p-1$), and use Chinese
remaindering to recover the exponents of each nonzero term in $f$.

Observe that the list $L$ can be implemented in any convenient way
according to the details of the parallel implementation; it serves only
as an unordered accumulator of coefficient-exponent-prime tuples.
The output polynomial $f$ could also be considered as an unordered
accumulator, followed by another parallel sort before the final return
statement.

\subsection{Parallel complexity analysis}

We state the parameterized running time of the algorithm as follows:

\begin{theorem}\label{thm:rtime}
  Given bounds $T,D$ on the sparsity and degree of an unknown
  polynomial $f\in\ZZ[x_1,\ldots,x_n]$, and parameters $k,\ell,Q$,
  Algorithm \ref{proc:interp} has worst-case running time
  \begin{align*}
    \softoh\bigg(&\log \ell +
    n \log D \log Q +
    k T \log Q \\
    &+ \left(\tfrac{1}{m}\ell n \log D\right)\left(
      n + \log D + k T \log Q
    \right)
  \bigg).
  \end{align*}
\end{theorem}

\begin{proof}
  The computation of $\mu$ at the beginning incurs the $\softoh(\log \ell)$
  cost in the complexity, which for our ultimate choices of $\ell$ will
  never actually dominate the complexity.

  The first loop executes 
  $\lceil \mu/m \rceil \in \softoh(1 + \tfrac{1}{m}\ell n \log D)$
  times in each thread.
  Computing the powers of $D$ modulo $p_i$ on Step~\ref{step:Dpow}
  requires $\softoh(\log D + n\log p_i)$ bit operations. The subsequent
  evaluations using dense arithmetic on Step~\ref{step:eval} will
  usually dominate the complexity of the entire algorithm, as each costs
  $\softoh(p_i \log q)$, which is $\softoh(k T \log Q)$.
  (The addition of this term makes the $\log p_i$ factor in the cost of
  Step~\ref{step:Dpow} become $\softoh(1)$.)

  Both parallel sorts are on lists of size at most $\mu T$, which means
  their cost of $\softoh(\tfrac{1}{m} \ell n T \log D)$ does not
  dominate the complexity.

  The final for loop executes $O(1 + \tfrac{1}{m} \mu T)$ times in each
  thread, but the nested if statement can only be triggered $O(T)$ times
  overall. Within it, the most expensive step is computing each
  $\alpha^{-E_c} \bmod q$, requiring $\softoh(n\log D \log Q)$
  bit operations. This contributes only
  $\softoh(n\log D \log Q)$ to the overall complexity, as
  the term $\frac{1}{m}T n\log D \log Q$ is already dominated by the
  parallel cost of Step~\ref{step:eval}.
\end{proof}

The key feature of procedure \ref{proc:interp} is that its potential
parallel speedup, from the previous theorem, is a factor of 
$O(\ell n \log D)$, depending on the number of parallel processors $m$
that are available. This exceeds the $O(n)$ parallel speedup of previous
approaches, and means that our algorithm should scale better to a large
number of processors when the number of variables and/or degree are
sufficiently large.

\subsection{Correctness and probability analysis}

We first use prior work to prove the bounds necessary to ensure
correctness with provably high probability. 

\begin{theorem}\label{thm:bounds}
  If the parameters $k,\ell,Q$ and bounds $D,T,H$ satisfy 
  \begin{align*}
    Q &\ge \max\left(2H, \tfrac{1}{4}(\ell n T \lg D)^2 D\right),  \\
    k &\ge \max(21,\lceil20 n \ln D\rceil), \text{and} \\
    \ell &\ge 2,
  \end{align*}
  then with
  probability at least $1/2$, procedure \ref{proc:interp} correctly
  computes the coefficients and exponents of $f\in\ZZ[x_1,\ldots,x_n]$.
\end{theorem}
\begin{proof}
  The condition $Q \ge 2H$ guarantees that positive or negative
  coefficients whose absolute value is at most $H$ will all be distinct
  modulo $q$, for any prime $q\ge Q$ as chosen in the algorithm.

  Because $\lambda = kT$, Lemma 8 from \cite{AGR13}
  tells us that, for each $p_i$, the probability that more than $T/3$
  terms collide modulo $z^{p_i}-1$ is at most $1/4$.
  Since the most number of terms that could collide is $T$, this means
  the \emph{expected} number of collisions, for each $p_i$, is at
  most $T/2$.

  The total number of nonzero coefficients in all polynomials $f_i$ examined on
  step~\ref{step:growL} is at most $\mu T \le \ell n T \lg D$. Many of
  these correspond to single terms in $f$ itself, but some will be
  collisions of terms in $f$. Using the reasoning of Lemma 4.1 in
  \cite{AGR14}, and the proof of Theorem 3.1 from \cite{GR11a},
  every coefficient of a single term in $f$, or a collision that appears
  in any of the $f_i$'s, is distinct modulo $q$, due to the bound on
  $Q$ and because $q \ge Q \ge \tfrac{1}{2}(\mu T)^2 \lg D$.

  We conclude that, with probability at least $1/2$, each term in $f$
  appears un-collided in at least $\mu/2$ of the polynomials $f_i$, and
  furthermore that these coefficients are all distinct and are the only
  ones that repeat in the list $L$. This guarantees that the correct
  coefficients and exponents are recovered in the final loop, and the
  algorithm outputs the correct interpolated polynomial.
\end{proof}

Applying the bounds in Theorem~\ref{thm:bounds} to the analysis in
Theorem~\ref{thm:rtime} gives the provable complexity bound for the
algorithm as stated in \eqref{eqn:rtg}.

As usual, the probability of success in either the provable or the
heuristic version of the algorithm (described below)
can be increased arbitrarily high by
running the same algorithm repeatedly and choosing the most common
polynomial returned among all runs to be the most likely candidate for
$f$.

\subsection{Heuristic approach}

The heuristic version of this procedure is simply to choose
appropriate constants for the parameters $k$ and $\ell$, with the
intuition that there can be some trade-off between the size of each prime
(governed by $k$) and the number of chosen primes (governed by $\ell$).
Furthermore, the bound on $Q$ required in theory to obtain
diversification between the coefficients in $f$ and in any collisions is
unnecessarily high for most ``typical'' polynomials. There do exist
pathological counterexamples, but they require the degrees of many terms
to be equivalent modulo $q$. As the prime $q$ is also chosen randomly in
our approach, we have a good indication that this heuristic approach
will work with high probability for a randomly-chosen sparse polynomial.
We state this as a conjecture, which is also backed by the experimental
evidence reported in Section~\ref{sec:exper} below.

\begin{claim}
  For any sufficiently large height bound $H$, and using $Q=2H$, 
  there exist constants
  $k,\ell\ge 1$ such that, for a polynomial $f\in\ZZ[x_1,\ldots,x_n]$
  chosen at random with at most $T$ nonzero terms, the probability that
  algorithm \ref{proc:interp} successfully interpolates $f$ is at least
  1/2.
\end{claim}

The heuristic complexity under this conjecture is stated in \eqref{eqn:rth}.

\subsection{Example}

 Consider an unknown bivariate
 polynomial $f(x,y)$ with 3 nonzero terms and degree less
 than 10.

The primes that we are going to use are going to be small for the sake of this 
example. The primes are  7, 13, and 17.

\begin{enumerate}

 \item Now we compute $f$ modulo $(z^p_i-1)$, we get:

  \begin{align*}
    f \bmod (z^{17}-1) &= 2z^{9}+7z^{8}+3z^{3}  \\
    f \bmod (z^{13}-1) &= 10z^{7}+2z^{4}        \\
    f \bmod (z^{7}-1) &= 3z^{6}+7z^{3}+2z       \\    
  \end{align*}

  We notice here that a collision happened with the prime 13. This won't affect our 
  calculation later because resulting coefficient $10$ doesn't appear anywhere else 
  (For this example, we are going to assume that we have a good coefficient if it 
  appears twice or more).
    
  Now we use these values to fill the list $L$. Every triple in $L$
  consists of the coefficient, the prime, and the exponent.
  
  $$
  \begin{array}{c|*{8}{c|}}
    \cline{2-9}
    \text{coefficient} & 2 & 7 & 3 & 10 & 2 & 3 & 7 & 2 \\
    \cline{2-9}
    \text{exponent} & 9 & 8 & 3 & 7 & 4 & 6 & 3 & 1     \\
    \cline{2-9}
    \text{prime} & 17 & 17 & 17 & 13 & 13 & 7 & 7 & 7    \\
    \cline{2-9}
  \end{array}
  $$
    
  We then sort $L$ based on coefficient values.

  $$
  \begin{array}{c|*{8}{c|}}
    \cline{2-9}
    \text{coefficient} & 10 & 7 & 7 & 3 & 3 & 2 & 2 & 2 \\
    \cline{2-9}
    \text{exponent} & 7 & 8 & 3 & 3 & 6 & 9 & 4 & 1      \\
    \cline{2-9}
    \text{prime} & 13 & 17 & 7 & 17 & 7 & 17 & 13 & 7       \\
    \cline{2-9}
  \end{array}
  $$
  
  We use the Chinese remainder theorem to get back the exponent that corresponds 
  to each coefficient.

  $$
    \left.
    \begin{array}{r@{\ }c@{\ }l}      
      e_1 \bmod 17 &=& 8 \\
      e_1 \bmod 7 &=& 3  \\
    \end{array}
    \right\}
    \Rightarrow e_1 = 59
  $$

  $$
    \left.
    \begin{array}{r@{\ }c@{\ }l}      
      e_2 \bmod 17 &=& 3 \\
      e_2 \bmod 7 &=& 6 \\
    \end{array}
    \right\}
    \Rightarrow e_2 = 20
  $$
    
  $$
    \left.
    \begin{array}{r@{\ }c@{\ }l}
      e_3 \bmod 17 &=& 9 \\
      e_3 \bmod 13 &=& 4 \\
      e_3 \bmod 7 &=& 1 \\
    \end{array}
    \right\}
    \Rightarrow e_3 = 43
  $$
    
  This results in the univariate polynomial
  $$f(z) = 3z^{20} + 2z^{43} + 7z^{59}.$$

  Finally, inverting the Kronecker map $f(z) = f(x,y^{10})$,
  we obtain
  $$f(x,y) = 3y^2 + 2x^3y^4 + 7x^9y^5.$$

\end{enumerate}

\iftoggle{article}{\FloatBarrier}{}

\section{Parallel implementation}
\label{sec:impl}

We completed a low-level implementation of Procedure \ref{proc:interp}
written in the C programming language. Our complete source code, as well
as the exact source we tested for the comparisons and benchmark problems
listed later, is available upon request by email. 

We give a few details
here on the choices of our implementation, in particular the libraries
that were utilized.

\subsection{FLINT for sparse and dense polynomial arithmetic}

The key advantage to the ``small primes'' approach which we employed is
the reliance on fast subroutines for dense polynomial arithmetic. The
experiments we ran always used a word-sized modulus $q$, and so the most
expensive computations involved computing with dense, low-degree
polynomials with word-sized modular coefficients.

FLINT (\url{http://flintlib.org/}) is a free, open-source C library for
fast number theoretic computations \cite{flint}. Our dense polynomial
arithmetic, which is the dominating cost both in theory and in practice
in our experiments, was performed using the \verb|nmod_poly| data type.

In order to store the result of sparse interpolation and complete the
correctness testing in our experiments, we also added rudimentary
support for sparse integer polynomials on top of FLINT. We created a new
data type, \verb|fmpz_sparse|, to represent sparse univariate
polynomials in $\ZZ[x]$ for testing purposes, using FLINT's multiple
precision type \verb|fmpz| as both the coefficient and exponent storage
for sparse polynomials. 

Note that using multiple-precision integers for
exponents is especially important, as we have \emph{not} yet completed
full multivariate polynomial support within FLINT. Instead, for the
purposes of our experiments, we always used a standard Kronecker
substitution to store $n$-variate, degree-$D$ 
multivariate polynomials as univariate sparse
polynomials with degree bounded by $D^n$. Employing a multiple-precision
data type allows for the largest possible degree and number of
variables, which is crucial as it is precisely this \emph{supersparse}
case in which our approach has the greatest potential parallel speedup.

\subsection{MPI for multi-processor parallelism}
 Message Passing Interface (MPI) allows us to parallelize our algorithm. 
 As stated earlier, since the slowest part of our algorithm involves calculating 
 the unknown polynomial $f$ modulo many polynomials $(x^{p_i}-1)$, we
 use MPI to perform each of these evaluations in parallel.

 The function \verb|MPI_Init|
 is called at the beginning of the program to spawn an arbitrary number
 of processes, as specified on the command line.
 Each of the allocated processes will be
 executing separately with separate copies of all variables in the original 
 process. All processes will have unique id numbers. The root process
 will have id number 0. By knowing processes id's we can separate what each
 process executes. 
 
 We used a master-slave model for our algorithm. The master 
 process evenly distributes how many primes each slave calculates. After getting 
 the primes, the slave process will compute the following for every prime:
 $\prod_{i=1}^{m} f_i(x) \bmod (x^{p_i}-1)$, then
 traverse the resulting univariate polynomial 
 and save all nonzero terms, along with the prime $p_i$, to an array.
 This array of triples is sent back to the master process, which later
 sorts all the concatenated evaluation arrays
 and uses Chinese remaindering to recover the full polynomial
 $f$, as described above.

 While our experiments were performed on a multi-core machine, and hence
 using a simpler threading library would have also worked, our goal in
 using MPI was to demonstrate the full parallel potential of this
 approach by explicitly detailing the inter-process communication.
 Furthermore, the MPI implementation could also be used without
 modification on heterogeneous clusters or other architectures besides
 multi-core.
  
\section{Experimental results}
\label{sec:exper}

We ran our tests on a machine using an Intel(R) Core(TM) i7-3930K CPU @ 3.20GHz 
simulating 12 hyper-threaded cores on 6 physical cores, with 32 GB of RAM. 
We used a Debian GNU system, running the ``unstable'' branch, 
with Linux kernel version 3.16.0-4-amd64. This is a bleeding-edge system
with the most current versions of all software available within the
Debian repositories.

\subsection{Determining the parameters $k$ and $\ell$}

The first task was to determine experimentally what kind of settings for
the parameters $k$ and $\ell$ would be appropriate for our heuristic
interpolation method. We found that $k=38$ and $\ell=2$ worked for a
wide range of problem sizes with almost no failures in the randomized
algorithm.

The only exceptions we found here were that when the bound on the number of
terms $T$ in the output was very small, even setting $k=38$ the number
of primes $p$ in the range $[\lambda,2\lambda]$ was simply not
sufficient to ensure a high success probability. In these small-size
extremes, the value of $k$ was increased to accommodate; in particular, 
we settled on $k=50$ and
$\ell=2$ when $T < 1000$, and when $T < 100$ we had to select 
$k \ge 10000/T$. Under these parameter settings, no failures were
observed in any of our experiments.

We comment that a smaller $k$ value and a
larger $\ell$ value would be preferable, because that 
would increase the number of primes $\mu$ and
hence the potential parallel speedup
for the algorithm, while decreasing the size $\lambda$ of each prime
$p_i$. However, we found that modestly larger values for $\ell$ did not
allow for $k$ to be reduced and consistently result in correct output.
We consider the exploration of some better balance between the
parameters $k$ and $\ell$ as future work.

\subsection{Single-threaded performance}

\begin{table*}[tbp]
\iftoggle{article}{\begin{adjustbox}{center}}{\begin{center}}

\begin{tabular}{*{6}{r}*{3}{|r}}
\multicolumn{1}{c}{$m$} &
\multicolumn{1}{c}{variables} &
\multicolumn{1}{c}{terms} &
\multicolumn{1}{c}{max-degree} &
\multicolumn{1}{c}{$\mu$} &
\multicolumn{1}{c}{$\lambda$} &
\multicolumn{1}{|c|}{Mathemagix} &
\multicolumn{1}{|c|}{Ours (single thread)} &
\multicolumn{1}{|c}{Ours (12 threads)} \\
\hline
1 & 20 & 3    & 40  & 14 & 50\,000  & 0.078 & 0.035 & 0.029 \\
2 & 20 & 9    & 80  & 18 & 5\,000   & 0.155 & 0.151 & 0.048 \\
3 & 20 & 27   & 120 & 18 & 6\,900   & 0.305 & 0.329 & 0.116 \\
4 & 20 & 81   & 160 & 18 & 4\,900   & 0.598 & 0.323 & 0.085 \\
5 & 20 & 243  & 200 & 17 & 9\,900   & 2.156 & 0.785 & 0.175 \\
6 & 20 & 729  & 240 & 15 & 39\,900  & 5.053 & 3.084 & 0.814 \\
7 & 20 & 2187 & 280 & 14 & 80\,050  & 13.333 & 8.714 & 2.225 \\
8 & 20 & 6561 & 320 & 13 & 321\,300 & 41.070 & 43.911 & 10.605 \\
\end{tabular}

\iftoggle{article}{\end{adjustbox}}{\end{center}}
\caption{Sparse interpolation benchmark times (in
seconds)\label{tab:comp}}
\end{table*}

\begin{figure}[tbp]
\includegraphics[width=\linewidth]{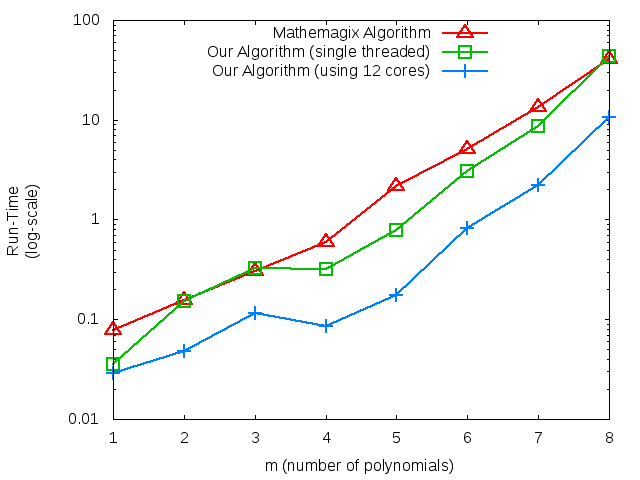}
\caption{Comparison With Mathemagix Sparse Interpolation Program\label{fig:comparison}}
\end{figure}

The first experiment was to test our algorithm 
without parallelization against
the efficient ``big primes'' implementation in Mathemagix 
(\url{http://www.mathemagix.org/}), as
reported in \cite{HL15}.
We downloaded the source code for the Mathemagix sparse
interpolation program, then ran it with $m=1, 2, 3, 4, 5, 6, 7, 8$ polynomials 
being multiplied. We kept each polynomial at 20 variables, with degree 40, 3 terms 
and coefficients up to $2^{30}$. Then we ran our algorithm with the same parameters. 
The results are shown in Taible~\ref{tab:comp} and summarized in
Figure~\ref{fig:comparison}. 
(We also
attempted a comparison against the \texttt{sinterp} function in Maple 2015,
but it was more than an order of magnitude slower in all experiments.)

Our non-parallelized algorithm seems to be on par with the Mathemagix
implementation for this range of problem sizes. As a comparison, and to
emphasize the main point of our paper, we also show the full parallel
speedup for the same problems in Figure~\ref{fig:comparison}. 

\subsection{Parallel speedup}

\begin{figure}[tbp]
\includegraphics[width=\linewidth]{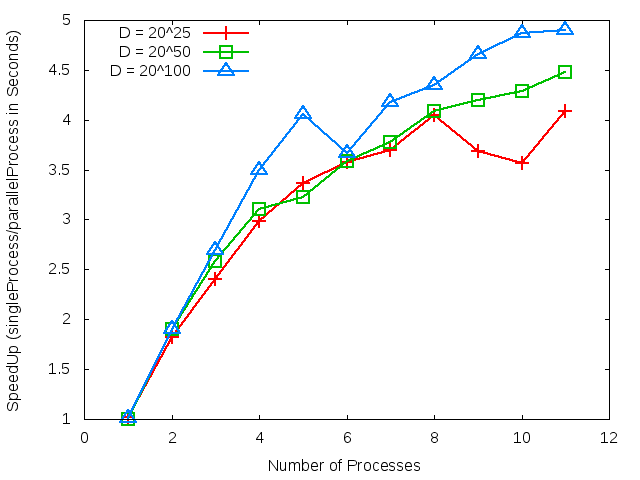}
\caption{Parallel speedup for our implementation\label{fig:speedup}}
\end{figure}

The second experiment was to test the parallel speedup for our implementation by 
varying the degree size $D$ and leaving $k=38$, $l=2$, partial $f_i$, and $m=6$ constant. $D$
was varied from $20^{25}$, $20^{50}$, $20^{100}$. This test was performed 3 separate times 
and the resultant data was calculated from median of the three trials. 
The results are seen in Figure~\ref{fig:speedup}. 

Figure~\ref{fig:speedup} shows a linear speedup increase as the number of 
parallel processes used gets closer to the physical number of cores on the machine.
Additionally, we can see the most significant parallel speedup occurs for the highest
degree tested. Recall that on our machine, there are only 6 physical
cores that are hyper-threaded to 12 virtual cores. Furthermore, when
running only six threads, the ``turbo mode'' clock rate is increased to
3.8GHz. This may help to
explain the dip in performance seen around $m=6$ parallel processes.

Observe also that the parallel speedup is best, and continues the
furthest, in the most extremely sparse case with very high degree.

Our parallel speedup is demonstrated in Figure~\ref{fig:speedup}.

\section{Conclusions and future work}
\label{sec:conc}

We have shown that the ``small primes'' sparse interpolation algorithm
is competitive with the state of the art, even without parallelization,
especially for very sparse problem instances. Furthermore, there is
greater potential for parallelism in the small primes technique. These
theoretical results are borne out in practice in our experimental results
compared to other available software implementations.

There is significantly more work to be done, however, before we might
suggest widespread adoption of our heuristic sparse interpolation
method. We would like to understand the theory behind the heuristic
approach in order to have a less \emph{ad hoc} way of determining the
parameters $k$ and $\ell$. On the other hand, our implementation could
be greatly enhanced with further experimentation on a wider range of
benchmark problems and incorporating true multivariate sparse polynomial
representations.

\section*{Acknowledgments}
This work was supported by the National Science Foundation under
award number \#1319994. We thank Andrew Arnold and the anonymous PASCO
2015 reviewers for their comments on an earlier draft of this paper.

\renewcommand{\bibpreamble}{\addcontentsline{toc}{section}{References}}
\bibliographystyle{plainnat}

\begin{thebibliography}{21}
\providecommand{\natexlab}[1]{#1}
\providecommand{\url}[1]{\texttt{#1}}
\expandafter\ifx\csname urlstyle\endcsname\relax
  \providecommand{\doi}[1]{doi: #1}\else
  \providecommand{\doi}{doi: \begingroup \urlstyle{rm}\Url}\fi

\bibitem[Arnold et~al.(2013)Arnold, Giesbrecht, and Roche]{AGR13}
Andrew Arnold, Mark Giesbrecht, and Daniel~S. Roche.
\newblock Faster sparse interpolation of straight-line programs.
\newblock In Vladimir~P. Gerdt, Wolfram Koepf, Ernst~W. Mayr, and Evgenii~V.
  Vorozhtsov, editors, \emph{Proc. Computer Algebra in Scientific Computing
  (CASC 2013)}, volume 8136 of \emph{Lecture Notes in Computer Science}, pages
  61--74. Springer, September 2013.
\newblock \doi{10.1007/978-3-319-02297-0_5}.

\bibitem[Arnold et~al.(2014)Arnold, Giesbrecht, and Roche]{AGR14}
Andrew Arnold, Mark Giesbrecht, and Daniel~S. Roche.
\newblock Sparse interpolation over finite fields via low-order roots of unity.
\newblock In \emph{Proceedings of the 39th International Symposium on Symbolic
  and Algebraic Computation}, ISSAC '14, pages 27--34, New York, NY, USA, 2014.
  ACM.
\newblock \doi{10.1145/2608628.2608671}.

\bibitem[Arnold et~al.(2015)Arnold, Giesbrecht, and Roche]{AGR15}
Andrew Arnold, Mark Giesbrecht, and Daniel~S. Roche.
\newblock Faster sparse multivariate polynomial interpolation of straight-line
  programs.
\newblock \emph{CoRR}, abs/1412.4088, 2015.
\newblock URL \url{http://arxiv.org/abs/1412.4088}.

\bibitem[Avenda{\~n}o et~al.(2006)Avenda{\~n}o, Krick, and Pacetti]{AKP06}
Mart{\'{\i}}n Avenda{\~n}o, Teresa Krick, and Ariel Pacetti.
\newblock Newton-{H}ensel interpolation lifting.
\newblock \emph{Found. Comput. Math.}, 6\penalty0 (1):\penalty0 81--120, 2006.
\newblock ISSN 1615-3375.
\newblock \doi{10.1007/s10208-005-0172-3}.

\bibitem[Ben-Or and Tiwari(1988)]{BT88}
Michael Ben-Or and Prasoon Tiwari.
\newblock A deterministic algorithm for sparse multivariate polynomial
  interpolation.
\newblock In \emph{Proceedings of the twentieth annual ACM symposium on Theory
  of computing}, STOC '88, pages 301--309, New York, NY, USA, 1988. ACM.
\newblock ISBN 0-89791-264-0.
\newblock \doi{10.1145/62212.62241}.

\bibitem[Blahut(1979)]{Bla79}
R.E. Blahut.
\newblock Transform techniques for error control codes.
\newblock \emph{IBM Journal of Research and Development}, 23\penalty0
  (3):\penalty0 299--315, May 1979.
\newblock \doi{10.1147/rd.233.0299}.

\bibitem[Boyer et~al.(2012)Boyer, Comer, and Kaltofen]{BCK12}
Brice Boyer, Matthew~T. Comer, and Erich~L. Kaltofen.
\newblock Sparse polynomial interpolation by variable shift in the presence of
  noise and outliers in the evaluations.
\newblock In \emph{Electr. Proc. Tenth Asian Symposium on Computer Mathematics
  (ASCM 2012)}, 2012.

\bibitem[Cuyt and Lee(2008)]{CL08}
Annie Cuyt and {Wen-shin} Lee.
\newblock A new algorithm for sparse interpolation of multivariate polynomials.
\newblock \emph{Theoretical Computer Science}, 409\penalty0 (2):\penalty0
  180--185, 2008.
\newblock ISSN 0304-3975.
\newblock \doi{10.1016/j.tcs.2008.09.002}.
\newblock Symbolic-Numerical Computations.

\bibitem[Garg and Schost(2009)]{GS09}
Sanchit Garg and {\'E}ric Schost.
\newblock Interpolation of polynomials given by straight-line programs.
\newblock \emph{Theoretical Computer Science}, 410\penalty0 (27-29):\penalty0
  2659--2662, 2009.
\newblock ISSN 0304-3975.
\newblock \doi{10.1016/j.tcs.2009.03.030}.

\bibitem[Giesbrecht and Roche(2011)]{GR11a}
Mark Giesbrecht and Daniel~S. Roche.
\newblock Diversification improves interpolation.
\newblock In \emph{Proceedings of the 36th international symposium on Symbolic
  and algebraic computation}, ISSAC '11, pages 123--130, New York, NY, USA,
  2011. ACM.
\newblock ISBN 978-1-4503-0675-1.
\newblock \doi{10.1145/1993886.1993909}.

\bibitem[Grenet et~al.(2015)Grenet, \Hoeven{van der Hoeven}, and Lecerf]{GHL15}
Bruno Grenet, Joris \Hoeven{van der Hoeven}, and Gr{\'e}goire Lecerf.
\newblock Randomized root finding over finite {FFT}-fields using tangent
  {G}raeffe transforms.
\newblock In \emph{Proc. 40th International Symposium on Symbolic and Algebraic
  Computation}, ISSAC '15, page to appear, 2015.

\bibitem[Hart et~al.(2013)Hart, Johansson, and Pancratz]{flint}
W.~Hart, F.~Johansson, and S.~Pancratz.
\newblock {FLINT}: {F}ast {L}ibrary for {N}umber {T}heory, 2013.
\newblock Version 2.4.0, \url{http://flintlib.org}.

\bibitem[\Hoeven{van der Hoeven} and Lecerf(2015)]{HL15}
Joris \Hoeven{van der Hoeven} and Gr{\'e}goire Lecerf.
\newblock Sparse polynomial interpolation in practice.
\newblock \emph{ACM Commun. Comput. Algebra}, 48\penalty0 (3/4):\penalty0
  187--191, February 2015.
\newblock \doi{10.1145/2733693.2733721}.

\bibitem[Javadi and Monagan(2010)]{JM10}
Seyed Mohammad~Mahdi Javadi and Michael Monagan.
\newblock Parallel sparse polynomial interpolation over finite fields.
\newblock In \emph{Proceedings of the 4th International Workshop on Parallel
  and Symbolic Computation}, PASCO '10, pages 160--168, New York, NY, USA,
  2010. ACM.
\newblock ISBN 978-1-4503-0067-4.
\newblock \doi{10.1145/1837210.1837233}.

\bibitem[Kaltofen and Lee(2003)]{KL03}
Erich Kaltofen and {Wen-shin} Lee.
\newblock Early termination in sparse interpolation algorithms.
\newblock \emph{Journal of Symbolic Computation}, 36\penalty0 (3-4):\penalty0
  365--400, 2003.
\newblock ISSN 0747-7171.
\newblock \doi{10.1016/S0747-7171(03)00088-9}.
\newblock ISSAC 2002.

\bibitem[Kaltofen and Yagati(1989)]{KY89}
Erich Kaltofen and Lakshman Yagati.
\newblock Improved sparse multivariate polynomial interpolation algorithms.
\newblock In P.~Gianni, editor, \emph{Symbolic and Algebraic Computation},
  volume 358 of \emph{Lecture Notes in Computer Science}, pages 467--474.
  Springer Berlin / Heidelberg, 1989.
\newblock \doi{10.1007/3-540-51084-2_44}.

\bibitem[Kaltofen et~al.(1990)Kaltofen, Lakshman, and Wiley]{KLW90}
Erich Kaltofen, Y.~N. Lakshman, and John-Michael Wiley.
\newblock Modular rational sparse multivariate polynomial interpolation.
\newblock In \emph{Proceedings of the international symposium on Symbolic and
  algebraic computation}, ISSAC '90, pages 135--139, New York, NY, USA, 1990.
  ACM.
\newblock ISBN 0-201-54892-5.
\newblock \doi{10.1145/96877.96912}.

\bibitem[Kaltofen(2010)]{Kal10a}
Erich~L. Kaltofen.
\newblock Fifteen years after {DSC} and {WLSS2}: {W}hat parallel computations
  {I} do today [invited lecture at {PASCO} 2010].
\newblock In \emph{Proceedings of the 4th International Workshop on Parallel
  and Symbolic Computation}, PASCO '10, pages 10--17, New York, NY, USA, 2010.
  ACM.
\newblock ISBN 978-1-4503-0067-4.
\newblock \doi{10.1145/1837210.1837213}.

\bibitem[Kaltofen et~al.(2011)Kaltofen, Lee, and Yang]{KLY11}
Erich~L. Kaltofen, Wen-shin Lee, and Zhengfeng Yang.
\newblock Fast estimates of hankel matrix condition numbers and numeric sparse
  interpolation.
\newblock In \emph{Proceedings of the 2011 International Workshop on
  Symbolic-Numeric Computation}, SNC '11, pages 130--136, New York, NY, USA,
  2011. ACM.
\newblock ISBN 978-1-4503-0515-0.
\newblock \doi{10.1145/2331684.2331704}.

\bibitem[Mansour(1995)]{Man95}
Yishay Mansour.
\newblock Randomized interpolation and approximation of sparse polynomials.
\newblock \emph{SIAM Journal on Computing}, 24\penalty0 (2):\penalty0 357--368,
  1995.
\newblock \doi{10.1137/S0097539792239291}.

\bibitem[Zippel(1990)]{Zip90}
Richard Zippel.
\newblock Interpolating polynomials from their values.
\newblock \emph{Journal of Symbolic Computation}, 9\penalty0 (3):\penalty0
  375--403, 1990.
\newblock ISSN 0747-7171.
\newblock \doi{10.1016/S0747-7171(08)80018-1}.
\newblock Computational algebraic complexity editorial.

\end{thebibliography}

\newcommand{\Gathen}{\relax}\newcommand{\Hoeven}{\relax}

\end{document}